\documentclass[12pt]{article}

\usepackage{amsmath}
\usepackage{amsthm}
\usepackage{amsfonts}
\usepackage{amssymb}
\usepackage{enumerate}
\usepackage{setspace}
\usepackage{comment}
\usepackage{graphicx}
\usepackage{float}
\usepackage{tikz}
\usepackage{calc}
\usepackage{pgfplots}
\usepackage{subfigure}
\usepackage{setspace}
\usepackage{amssymb,amsmath}

\newtheorem{theorem}{Theorem}[section]
\newtheorem{proposition}[theorem]{Proposition}
\newtheorem{lemma}[theorem]{Lemma}

\usepackage{graphicx}
\usepackage{graphics}
\usepackage{xcolor}
\usepackage{tabu}
\usepackage{textcomp}
\usepackage{mathtools}

\title{A PageRank Model for Player Performance Assessment in Basketball, Soccer and Hockey}
\author{Shael Brown\\Department of Mathematics and Statistics\\
Dalhousie University, 
Halifax, Nova Scotia, Canada B3H 3J5}
\date{}

\begin{document}

\maketitle

\begin{abstract}
In the sports of soccer, hockey and basketball the most commonly used statistics for player performance assessment are divided into two categories: offensive statistics  and defensive statistics. However, qualitative assessments of playmaking (for example making ``smart'' passes) are difficult to quantify. It would be advantageous to have available a single statistic that can emphasize the flow of a game, rewarding those players who initiate and contribute to successful plays more. In this paper we will examine a model based on Google's PageRank. Other papers have explored ranking teams, coaches, and captains but here we construct ratings and rankings for individual members on both teams that emphasizes initiating and partaking in successful plays and forcing defensive turnovers.

For a soccer/hockey/basketball game, our model assigns a node for each of the $n$ players who play in the game and a ``goal node''. Arcs between player nodes indicate a pass in the \textit{reverse} order (turnovers are dealt with separately). Every sport-specific situation (fouls, out-of-bounds, play-stoppages, turnovers, missed shots, defensive plays) is addressed, tailored for each sport. As well, some additional arcs are added in to ensure that the associated matrix is primitive (some power of the matrix has all positive entries) and hence there is a unique PageRank vector. The PageRank vector of the associated matrix is used to rate and rank the players of the game. 

To illustrate the model, data was taken from nine NBA games played between 2014 and 2016. The model applied to the data showed that this model did indeed provide the type of comprehensive statistic described in the introductory paragraph. Many of the top-ranked players (in the model) in a given game had some of the most impressive traditional stat-lines. However, from the model there were surprises where some players who had impressive stat-lines had lower ranks, and others who had less impressive stat-lines had higher ranks.
Overall, the model provides an alternate tool for player assessment in soccer, basketball and hockey. The model's ranking and ratings reflect more the flow of the game compared to traditional sports statistics.
\end{abstract}

\section{Background}
Google PageRank was created as the backbone to what is now the most influential search engine ever created \cite{Beyond}. Its purpose is to rank the importance of web pages when a user makes a query. The foundations of PageRank lie in Markov chain theory: given a finite set of states $S = \{s_{1},\ldots,s_{n}\}$, let $t_{i,j}$ be the probability of moving from state $s_{i}$ to state $s_{j}$ from time $k$ to $k+1$, which is independent of $k$. Let $T = \left(  t_{i,j} \right)$ denote the \textit{transition matrix} of the Markov chain. Provided that the matrix $T$ is \textit{primitive} (i.e.., $T^{m} > 0$ for some positive integer $m$), there is a unique \textit{stationary vector}, $\mathbf{v}$, such that $\mathbf{v}^{t} \mathbf{1} = 1$ and $T^{t}\mathbf{v}=\mathbf{v}$, where $\mathbf{1}$ is the vector in $\mathbb{R}^{n}$ consisting of all $1$'s (that is,  $\mathbf{v}$ is an eigenvector of $T^{t}$ with eigenvalue $1$ whose nonnegative entries sum to $1$) \cite{PR}. We call such a  vector the \textit{PageRank vector} of the Markov chain. (The lack of primitivity in Google's Markov model in general requires some alteration to the transition matrix in that case.) Each component of the PageRank vector is thought of as the \textit{rank} of that state (and an ordering of the states is derived from these values). 

There is a natural way to construct a Markov chain from a (finite) directed graph. The states of the chain are the nodes of the graph. If there are $n_{i,j}$ arcs from node $i$ to node $j$ and node $i$ has a total of $n_i$ outgoing arcs, then $T_{i,j}=\frac{n_{i,j}}{n_i}$. For a Markov chain derived from a directed graph, primitivity of the Markov chain corresponds to the existence of a positive integer $k$ such that there is a walk of length $k$ between any two nodes. In such a case there is a unique stationary vector for the associated Markov chain, which can be calculated from a linear system (see \cite{Beyond} for more details). We remark that it has been observed that the PageRank vector is fairly insensitive to small changes in the network involving only lowly-ranked nodes \cite{Sensitivity}.

Previous applications of PageRank to sports metrics usually address ranking either teams \cite{NFL, Rankingrankings, Cricket}, coaches \cite{Coaches}, or individual players on various teams \cite{Cricket, CricketBB}. A number of models have relied on underlying graph networks of games in their respective sports. In \cite{strategies} a directed graph representing the passes between the starting players on an individual soccer team was constructed and a PageRank vector was computed to highlight who were the most important players on the team based on ball reception; the network as well was analyzed to determine the strategies and weak-points of each team. In \cite{Team} every team (in some team sport involving passing) has a corresponding weighted directed graph representing passes and two additional nodes representing the other team's goal and missed shots. Batsmen and bowlers that face each other on separate teams are compared using a basic model much like the one that compares teams in the ``win-loss'' PageRank method for ranking teams \cite{CricketBB}. A weighted directed graph for players across many basketball teams is created in \cite{Basketball} where arcs exist only between players who played on the court together on the same team at some point and the weight of these arcs corresponds to how effective they were in playing together. Finally, in \cite{Soccer} a PageRank network is created for each individual soccer team where arcs indicate passes between players. 

None of these models allow for the effective comparison of any two players playing in the same game together (or even in different games), possibly of different positions or teams, using a PageRank method. 
More importantly, they don't emphasize playmaking ability, as opposed to pure offensive statistics, and that is exactly what we plan to do.

\section{The Model}

We create a directed graph (which we abbreviate as a {\textit{digraph}) representing the progression of play during a particular game -- whether soccer, basketball or hockey. For each of the $n$ players in the game there is a node, and the digraph contains one additional  \textit{goal node} (while our implementation of a goal node is not unique, the use of a ``missed shot'' node in \cite{Team} is redundant when we have both teams that are competing against each other in a game represented on the same graph). 
Since we desire a model that values playmaking, we must reverse the direction of most arcs that would result from a straightforward progression of play (such an idea was raised but not deeply pursued in a multi-team setting in \cite{Soccer}). For player $i$ and player $j$ (represented by node $i$ and node $j$, respectively) on the same team, whenever player $i$ passes to player $j$ we draw an arc from node $j$ to node $i$. However, if players $i$ and $j$ are on separate teams and player $i$ loses the ball/puck to player $j$ we draw an arc from node $i$ to node $j$. If player $i$ scores, the sport specific value of the score will be the number of arcs drawn from the goal node to node $i$ (for example an NBA 3-pointer would result in three arcs). All of our choices for arc direction ensures the flow of rank rewards playmaking. There will be more game-specific arcs, to be discussed below.

We initialize the digraph for a given game as follows. We draw arcs in both directions between each player node and the goal node, and in addition, we draw a loop from the goal node to itself. Outside of goal scoring, no further arc will be drawn to or from the goal node. If team $1$ has $n_{1}$ players and team $2$ has $n_{2}$ players ($n_{1}+n_{2}=n$) then in the corresponding \textit{game transition matrix} $T \in M_{n+1,n+1}(\mathbb{R})$ the first $n_{1}$ columns (and rows) of $T$ represent players on team 1, columns (and rows) $n_{1}+1$ to $n$ representing players on team 2, and the $(n+1)$th column and row  representing the goal node. Thus, before the game starts, $T_{n+1,i}=\frac{1}{n+1}$ for $1 \leq i \leq n+1$ and $T_{i,j}=0$ otherwise. The method of construction of the initial digraph ensures that any two nodes are connected by a path of length exactly two, and hence the corresponding Markov chain transition matrix will have each entry of $T^2$ nonnegative, making $T$ primitive, and therefore the Markov chain associated with the digraph has a unique PageRank vector.

We let $\mathbf{r} = (r_{1},\ldots,r_{n},r_{\mbox{g}})$ be the PageRank vector of the game transition matrix. For a number of reasons, to be listed, we shall rescale the values in the PageRank vector (such a process does not change the induced ordering of the players' ranks). One immediate issue with the model is the fact that the goal node may have different rank in each game depending on the number of players in the game, thus making the comparison of ranks of players in different games dependent on the rank of the goal node. We can scale the computed rank vector $r=(r_{1},r_{2},...,r_{g})$ by any scalar (the eigenspace corresponding to eigenvalue $1$ has dimension $1$, as the digraph, being primitive, is strongly connected). We standardize $r$ by defining the \textit{integrated playmaking metric} (IPM) of player $i$ in the game by
\[ IPM_{i} = 50 n \cdot \frac{r_{i}}{ \sum_{j=1}^{n}r_{j}} = \frac{50n \cdot r_{i}}{ 1 - r_{\mbox{g}}}.\] 
The choice of scaling is as follows: the denominator removes the effect of the goal node's rank, and the numerator ensures (a) that the IPM is insensitive to the number of players in the game and (b) provides values on a reasonable scale, between 0 and $1000$ (it is not hard to see that the average IPMs of all players in a game is $50$).
The IPMs can thus be used for meaningful comparison of players in different games.

We now return to how the digraph itself is built up in the three sport specific situations. We illustrate the process with basketball (the rules for soccer and hockey can be found in Appendices A and B respectively). Each bullet point is a play ``event'', with its subsequent descriptor the corresponding arc(s) to add to the digraph.

Basketball rules of implementation:
\begin{itemize}
\item Pass from player $i$ to player $j$.
	\begin{itemize}
		\item An arc from node $j$ to node $i$
	\end{itemize}
\item Player $i$ dispossesses player $j$. [This could include cases where player $i$ does not gain possession of the ball after dispossessing player $j$ -- for example a defensive touch leading to the ball being out of play or deflecting a pass still into play but away from its intended target.]
	\begin{itemize}
		\item An arc from node $j$ to node $i$
	\end{itemize}
\item Player $i$ scores $n$ points where $1 \leq n \leq 4$.
	\begin{itemize}
		\item $n$ arcs from the goal node to node $i$
	\end{itemize}
\item Player $i$ shoots when being contested and defended by player $j$ and misses the net. [Same as player $j$ dispossessing player $i$, play resuming with the rebounding/inbounding player. This case includes the situation where player $j$ blocks player $i$.]
	\begin{itemize}
		\item An arc from node $i$ to node $j$
	\end{itemize}
\item Player $i$ shoots and misses the net under no pressure and the ball is rebounded by player $j$. [Same as player $j$ dispossessing player $i$.]
	\begin{itemize}
		\item An arc from node $i$ to player $j$
	\end{itemize}
\item Player $i$ fouls player $j$ and player $j$ makes at least one free throw.
	\begin{itemize}
		\item If player $j$ makes $n>0$ free throws then $n$ arcs are created from the goal node to node $j$
	\end{itemize}
\item Player $i$ fouls player $j$ and player $j$ makes zero free throws. [Same as player $j$ dispossessing player $i$ -- it was a ``smart'' foul.]
	\begin{itemize}
		\item An arc from node $j$ to node $i$
	\end{itemize}
\item Any stoppage of play that does not have to do with the game (i.e. a technical foul, fan interference, injury, altercation etc.). [Play is dead.]
	\begin{itemize}
		\item No arc drawn
	\end{itemize}
\item Player $i$ intercepts a pass from player $j$. [Same as player $i$ dispossessing player $j$.]
	\begin{itemize}
		\item An arc from node $j$ to node $i$
	\end{itemize}
\item Player $i$ touches the ball without having possession (for example the ball hits player $i$, or a ``pinball'' play).
	\begin{itemize}
		\item No arc drawn
	\end{itemize}
\item Any unforced turnover by player $i$.
	\begin{itemize}
		\item No arcs drawn
	\end{itemize}
\end{itemize}

We illustrate the process with a small example. Suppose that two basketball teams, the Reds and the Blues, are playing against each other in a ``3-on-3'' match (where all baskets are worth one point). We will denote the players on the Reds by $A$, $B$ and $C$ and those on the Blues by $D$, $E$ and $F$. Before the game begins we have the setup of the digraph shown in Figure~\ref{example1} (where $G$ stands for the goal node). In this and subsequent diagrams of digraphs, if there exists more than one arc between two nodes we will draw one arc but label it with the number of arcs that exist between the two nodes. 
\pagebreak

\begin{figure}[ht]
\begin{center}
\includegraphics[scale=0.75]{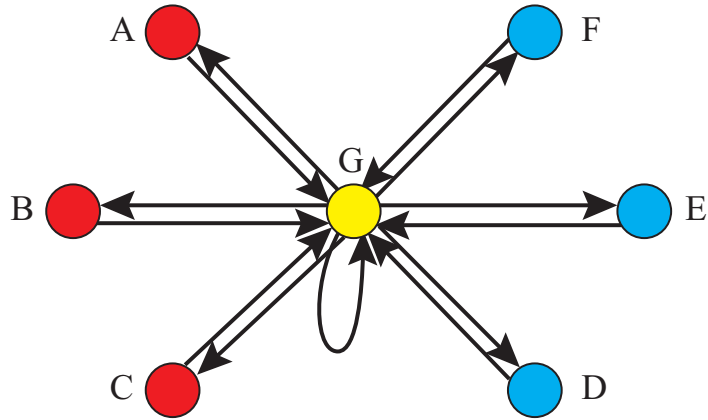}
\caption{Small example's initial digraph (before play).}
\label{example1}
\end{center}
\end{figure}

Now suppose we have the following sequence of plays in the game: (where ``$\rightarrow$'' represents the movement of the ball between nodes and ``0'' is the end of play sequence symbol):

\noindent $A \rightarrow B \rightarrow A \rightarrow F \rightarrow G$\\
$D \rightarrow F \rightarrow E \rightarrow F \rightarrow D \rightarrow C \rightarrow B  \rightarrow C \rightarrow A \rightarrow C  \rightarrow B \rightarrow A \rightarrow G$\\
$D \rightarrow C \rightarrow A \rightarrow C \rightarrow B \rightarrow A \rightarrow G$\\
$D \rightarrow F \rightarrow 0 \rightarrow B \rightarrow C \rightarrow A \rightarrow G$\\
$D \rightarrow F \rightarrow E \rightarrow F \rightarrow D \rightarrow G$\\
$A \rightarrow B \rightarrow F \rightarrow G$

\vspace{0.5in}
\noindent After these sequences of plays our updated network becomes:\\
\pagebreak
\begin{figure}[ht]
\begin{center}
\includegraphics[scale=0.75]{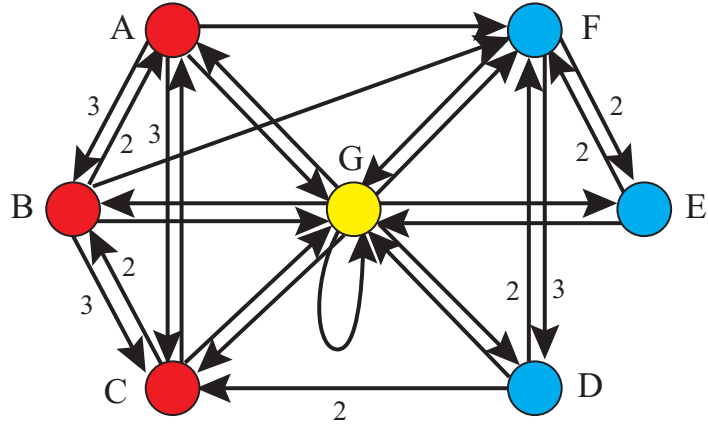}
\caption{Small example updated}
\label{model}
\end{center}
\end{figure}

\noindent The adjacency matrix of the digraph (whose $(i,j)$--th entry is the number of arcs in the digraph from node $i$ to node $j$) is
\[\begin{bmatrix}
0 & 3 & 3 & 0 & 0 & 1 & 1\\
2 & 0 & 3 & 0 & 0 & 1 & 1\\
2 & 2 & 0 & 0 & 0 & 0 & 1\\
0 & 0 & 2 & 0 & 0 & 2 & 1\\
0 & 0 & 0 & 0 & 0 & 2 & 1\\
0 & 0 & 0 & 3 & 2 & 0 & 1\\
4 & 1 & 1 & 2 & 1 & 3 & 1
\end{bmatrix},\]
and so the transition matrix for this game is 
\[T=\begin{bmatrix}
0 & 2/7 & 2/5 & 0 & 0 & 0 & 4/13\\
3/8 & 0 & 2/5 & 0 & 0 & 0 & 1/13\\
3/8 & 3/7 & 0 & 2/5 & 0 & 0 & 1/13\\
0 & 0 & 0 & 0 & 0 & 1/2 & 2/13\\
0 & 0 & 0 & 0 & 0 & 1/3 & 1/13\\
1/8 & 1/7 & 0 & 2/5 & 2/3 & 0 & 3/13\\
1/8 & 1/7 & 1/5 & 1/5 & 1/3 & 1/6 & 1/13
\end{bmatrix}.\]
The calculated IPMs of the players in the system are as follows:

\begin{table}[ht]
\begin{center}
\begin{tabular}{| c | c | c |} \hline
Player & Team & IPM \\ \hline\hline
C & Reds & 64.66\\ \hline
F & Blues & 60.38\\ \hline
A & Reds & 58.79\\ \hline
B & Reds & 52.39\\ \hline
D & Blues & 39.17\\ \hline
E & Blues & 24.61\\ \hline
\end{tabular}
\end{center}
\label{default}
\end{table}

While player $A$ scored the most goals in the game and player $F$ is tied with player $A$ for the most points (sum of goals and assists), we see that in fact player $C$, who had only 1 assist and no goals, had the highest IPM. However, a cursory examination of the plays clearly shows how integral player $C$ was in the game as a playmaker. The example shows how a player whose contribution to the game might be ignored under the usual stats lines receives their well deserved acknowledgement under the proposed model.


Before we continue with some experimental results it is natural to ask how this model fits with our intuition of evaluating the performance of athletes. 

We observe first that the lowest possible IPM of any player in a game of $n$ players is no more than $50$, since if the smallest IPM of any player in the game is more than 50 we would contradict the fact that the average IPM of all players in the game is always 50. 
The following propositions consider the spacing of IPMs in a game.

\begin{proposition}
If there are $k$ starters (out of $n$ players) in a given game with a cumulative starter IPM of $R$ then there must exist a bench player with a IPM that is at most $\frac{n(R-50k)}{k(n-k)}$ distance from the IPM of some starter.
\end{proposition}
\begin{proof}
We first note that the average bench player IPM must be $\frac{50n-R}{n-k}$ as there are $n-k$ bench players and the total sum of all the IPMs is $50n$. Clearly the worst case scenario is if all starters have a IPM of $\frac{R}{k}$ and all bench players have a IPM of $\frac{50n-R}{n-k}$ (as otherwise there would have to be one bench player with IPM greater than $\frac{50n-R}{n-k}$ or one starter with IPM less than $\frac{R}{k}$). Thus, in this worst case scenario the distance between any starter and any bench player is $\frac{R}{k} - \frac{50n-R}{n-k} =  \frac{Rn-Rk-50nk+Rk}{k(n-k)} = \frac{n(R-50k)}{k(n-k)}$.
\end{proof}

\begin{proposition}
If there are $n$ players in a game then there must exist at least two players who have IPMs within $\frac{25n}{n-2}$ of each other.
\end{proposition}
\begin{proof}
Clearly 0 is a lower bound on the IPM of a player in a given game and upper bound is $50n$ (the sum of all the IPMs). Thus, we may partition this range into $n-2$ equal length intervals, each of length $\frac{50n}{n-2}$ using $n-1$ of the $n$ total players. We must have that at least three players must be within (or on the boundary of) one of the $n-2$ intervals, meaning that two of their IPMs must be in the same half interval, making the difference of their IPMs no more than $\frac{25n}{n-2}$.
\end{proof}

The above two results show that there has to be some bench player of non-negligible importance to the game, and that some players have to have IPMs that are ``somewhat'' close to each other. The first result fits with the widely accepted notion that bench play is a component to the success of a basketball team.

Finally, we may be interested in how the number of goals scored by players on different teams affects their IPMs. As an illustration, suppose we have only two players, $P_1$ and $P_2$ in the system, each on separate teams, with no interaction between them and $P_1$ scores $g_1$ goals and $P_2$ scores $g_2$ goals. Then if $g_1 > g_2$ then the IPM of $P_1$ is greater than that of $P_2$. This follows as it is clear that the IPM of $P_1$ is $\frac{50nr_{g}(g_1 +1)}{(g_1 + g_2 +2)(1-r_g)}$ and the IPM of $P_2$ is $\frac{50nr_{g}(g_2 +1)}{(g_1 + g_2 +2)(1-r_g)}$ where $r_g$ is the rank of the goal node.

\section{Sample Data Analysis}

We now apply our model to some real-life basketball games.
The specific games used were 
\begin{itemize}
\item Chicago Bulls vs. San Antonio Spurs (November 30th 2015), 
\item Golden State Warriors vs. Cleveland Cavaliers (January 18th 2016), 
\item San Antonio Spurs vs.  Brooklyn Nets (December 3rd 2014), 
\item Chicago Bulls vs. Charlotte Hornets (December 3rd 2014), 
\item Los Angeles Lakers vs.  Golden State Warriors (November 1st 2014), 
\item Toronto Raptors vs. Cleveland Cavaliers (December 9th 2014), 
\item Golden State Warriors vs. Cleveland Cavaliers (December 25th 2015), 
\item Chicago Bulls vs.  Oklahoma City Thunder (December 25th 2015), and 
\item Washington Wizards vs.  Cleveland Cavaliers (November 26th 2014). 
\end{itemize}
The games are identified by the teams playing, date, score and winning team. 
In each game, plays were manually transcribed, and the resulting transition matrices and PageRank vectors were calculated. 
The nine tables list, in decreasing order, the IPMs of each of the players (rounded to the nearest hundredth) in all nine NBA basketball games from which data was taken, followed by each player's points (P), assists (A), rebounds (R), steals (S), turnovers (T) and field goal percentage (FG\%), all accessed from nba.com.\\ 


\noindent Chicago Bulls vs. San Antonio Spurs, November 30th 2015 (Bulls win 92-89)
\begin{table}[ht]
\begin{center}
\begin{tabular}{| l | l | c | r | r | r | r | r | r |} \hline

Player name & Team & IPM & P & A & R & S & T & FG\% \\ \hline \hline
Parker & Spurs & 103.66 & 13 & 9 & 1 & 0 & 0 & 50\\ \hline
Rose & Bulls & 83.68 & 11 & 6 & 4 & 1 & 1 & 29\\ \hline
Duncan & Spurs & 74.05 & 6 & 3 & 12 & 0 & 2 & 43\\ \hline
Gasol & Bulls & 72.30 & 18 & 4 & 13 & 1 & 1 & 33\\ \hline
Leonard & Spurs & 63.62 & 25 & 3 & 8 & 2 & 2 & 77\\ \hline
Aldridge & Spurs & 59.99 & 21 & 0 & 12 & 0 & 2 & 55\\ \hline
Noah & Bulls & 58.32 & 8 & 7 & 11 & 0 & 0 & 67\\ \hline
Green & Spurs & 57.69 & 9 & 1 & 4 & 2 & 1 & 30\\ \hline
Butler & Bulls & 49.23 & 14 & 3 & 3 & 1 & 5 & 55\\ \hline
Mirotic & Bulls & 46.55 & 8 & 2 & 5 & 0 & 2 & 38\\ \hline
Ginobli & Spurs & 44.89 & 4 & 2 & 1 & 1 & 1 & 25\\ \hline
Diaw & Spurs & 44.82 & 5 & 0 & 6 & 0 & 1 & 40\\ \hline
Mills & Spurs & 40.31 & 4 & 1 & 0 & 0 & 0 & 25\\ \hline
Moore & Bulls & 33.42 & 6 & 1 & 2 & 0 & 1 & 50\\ \hline
West & Spurs & 31.02 & 2 & 1 & 3 & 0 & 0 & 20\\ \hline
Snell & Bulls & 27.48 & 11 & 1 & 6 & 0 & 0 & 80\\ \hline
McDermott & Bulls & 24.97 & 12 & 0 & 3 & 0 & 0 & 42\\ \hline
Gibson & Bulls & 20.86 & 4 & 1 & 4 & 1 & 0 & 40\\ \hline
Anderson & Spurs & 13.13 & 0 & 0 & 0 & 0 & 0 & 0\\ \hline

\end{tabular}
\end{center}
\label{default}
\end{table}%

\pagebreak

\noindent Golden State Warriors vs. Cleveland Cavaliers, January 18th 2016 (Warriors win 132-98)
\begin{table}[ht]
\begin{center}
\begin{tabular}{| l | l | c | r | r | r | r | r | r |} \hline

Player name & Team & IPM & P & A & R & S & T & FG\%\\ \hline \hline
Curry & Warriors & 122.91 & 35 & 4 & 5 & 3 & 1 & 67\\ \hline
Green & Warriors & 109.34 & 16 & 10 & 7 & 0 & 1 & 50\\ \hline
Dellavedova & Cavaliers & 109.26 & 11 & 6 & 1 & 0 & 1 & 50\\ \hline
James & Cavaliers & 91.43 & 16 & 5 & 5 & 1 & 3 & 44\\ \hline
Irving & Cavaliers & 78.75 & 8 & 3 & 5 & 0 & 2 & 27\\ \hline
Barnes & Warriors & 66.94 & 12 & 0 & 2 & 0 & 2 & 50\\ \hline
Livingston & Warriors & 62.80 & 4 & 0 & 2 & 0 & 2 & 67\\ \hline
Love & Cavaliers & 60.54 & 3 & 2 & 6 & 0 & 1 & 20\\ \hline
Iguodala & Warriors & 60.24 & 20 & 5 & 3 & 0 & 1 & 88\\ \hline
Bogut & Warriors & 58.61 & 4 & 0 & 6 & 0 & 0 & 67\\ \hline
Varejao & Cavaliers & 50.63 & 5 & 3 & 4 & 1 & 1 & 50\\ \hline
Shumpert & Cavaliers & 39.84 & 10 & 0 & 2 & 0 & 3 & 67\\ \hline
Barbosa & Warriors & 39.52 & 8 & 4 & 1 & 2 & 0 & 50\\ \hline
Thompson & Warriors & 39.05 & 15 & 2 & 1 & 1 & 1 & 45\\ \hline
Clark & Warriors & 38.24 & 6 & 2 & 2 & 0 & 0 & 29\\ \hline
Ezeli & Warriors & 29.79 & 4 & 0 & 2 & 0 & 2 & 67\\ \hline
Mozgov & Cavaliers & 28.84 & 6 & 3 & 0 & 0 & 1 & 50\\ \hline
Smith & Cavaliers & 26.99 & 14 & 1 & 2 & 1 & 1 & 67\\ \hline
Thompson & Cavaliers & 23.73 & 2 & 0 & 2 & 0 & 0 & 0\\ \hline
J Thompson & Warriors & 21.89 & 1 & 1 & 3 & 0 & 0 & 1\\ \hline
Cunningham & Cavaliers & 20.39 & 9 & 1 & 3 & 0 & 0 & 60\\ \hline
Jefferson & Cavaliers & 18.47 & 6 & 0 & 3 & 0 & 1 & 100\\ \hline
Jones & Cavaliers & 17.43 & 8 & 1 & 0 & 0 & 1 & 60\\ \hline
Speights & Warriors & 17.22 & 4 & 0 & 1 & 0 & 0 & 25\\ \hline
Rush & Warriors & 17.17 & 3 & 1 & 3 & 0 & 0 & 33\\ \hline

\end{tabular}
\end{center}
\label{default}
\end{table}%

\pagebreak

\noindent San Antonio Spurs vs. Brooklyn Nets, December 3rd 2014 (Nets win 95-93)
\begin{table}[ht]
\begin{center}
\begin{tabular}{| l | l | c | r | r | r | r | r | r |} \hline

Player name & Team & IPM & P & A & R & S & T & FG\%\\ \hline \hline
Williams & Nets & 111.40 & 17 & 9 & 3 & 0 & 2 & 40\\ \hline
Teletovic & Nets & 90.99 & 26 & 2 & 15 & 0 & 0 & 69\\ \hline
Duncan & Spurs & 81.11 & 14 & 1 & 17 & 0 & 2 & 28\\ \hline
Parker & Spurs & 79.92 & 9 & 6 & 1 & 0 & 3 & 50\\ \hline
Ginobli & Spurs & 68.23 & 15 & 5 & 6 & 1 & 0 & 46\\ \hline
Green & Spurs & 65.77 & 20 & 2 & 10 & 2 & 0 & 50\\ \hline
Lopez & Nets & 62.01 & 16 & 3 & 16 & 0 & 0 & 35\\ \hline
Leonard & Spurs & 56.59 & 12 & 1 & 13 & 1 & 0 & 25\\ \hline
Joseph & Spurs & 55.32 & 7 & 3 & 3 & 0 & 0 & 38\\ \hline
Johnson & Nets & 53.48 & 8 & 2 & 5 & 1 & 1 & 25\\ \hline
Jack & Nets & 47.11 & 8 & 3 & 1 & 0 & 2 & 40\\ \hline
Diaw & Spurs & 45.72 & 0 & 3 & 2 & 0 & 2 & 0\\ \hline
Bonner & Spurs & 35.93 & 7 & 0 & 1 & 0 & 0 & 30\\ \hline
Bogdanovic & Nets & 32.72 & 14 & 0 & 8 & 0 & 2 & 50\\ \hline
Baynes & Spurs & 17.78 & 4 & 1 & 4 & 1 & 0 & 40\\ \hline
Anderson & Nets & 14.49 & 2 & 1 & 1 & 0 & 2 & 20\\ \hline
Belinelli & Spurs & 11.89 & 5 & 1 & 1 & 0 & 1 & 67\\ \hline
Jordan & Nets & 9.83 & 2 & 0 & 2 & 1 & 0 & 100\\ \hline
Plumlee & Nets & 9.73 & 2 & 0 & 3 & 0 & 1 & 33\\ \hline

\end{tabular}
\end{center}
\label{default}
\end{table}%

\pagebreak

\noindent Chicago Bulls vs. Charlotte Hornets, December 3rd 2014 (Bulls win 102-95)
\begin{table}[ht]
\begin{center}
\begin{tabular}{| l | l | c | r | r | r | r | r | r |} \hline

Player name & Team & IPM & P & A & R & S & T & FG\% \\ \hline \hline
Walker & Hornets & 110.08 & 23 & 4 & 5 & 1 & 0 & 39\\ \hline
Rose & Bulls & 85.77 & 15 & 5 & 2 & 0 & 2 & 42\\ \hline
Gasol & Bulls & 84.94 & 19 & 3 & 15 & 0 & 2 & 37\\ \hline
Noah & Bulls & 77.80 & 14 & 7 & 10 & 1 & 2 & 67\\ \hline
Mirotic & Bulls & 73.59 & 11 & 1 & 2 & 0 & 1 & 50\\ \hline
Stephenson & Hornets & 64.26 & 20 & 4 & 8 & 1 & 4 & 50\\ \hline
Zeller & Hornets & 63.21 & 12 & 2 & 8 & 0 & 0 & 45\\ \hline
Williams & Hornets & 59.65 & 6 & 0 & 3 & 1 & 0 & 40\\ \hline
Butler & Bulls & 52.88 & 15 & 5 & 2 & 2 & 1 & 45\\ \hline
Brooks & Bulls & 51.90 & 7 & 3 & 3 & 0 & 2 & 43\\ \hline
Hinrich & Bulls & 48.90 & 12 & 2 & 3 & 0 & 1 & 44\\ \hline
Roberts & Hornets & 37.93 & 3 & 3 & 1 & 0 & 0 & 13\\ \hline
Jefferson & Hornets & 35.36 & 13 & 2 & 7 & 0 & 0 & 38\\ \hline
Dunleavy & Bulls & 32.20 & 9 & 0 & 1 & 1 & 0 & 60\\ \hline
Henderson & Hornets & 30.55 & 10 & 1 & 4 & 0 & 1 & 50\\ \hline
Hairston & Hornets & 14.67 & 4 & 1 & 2 & 2 & 0 & 14\\ \hline
Snell & Bulls & 13.54 & 0 & 1 & 1 & 0 & 0 & 0\\ \hline
Biyombo & Hornets & 12.64 & 4 & 0 & 4 & 0 & 0 & 50\\ \hline
Pargo & Hornets & 0.14 & 0 & 0 & 0 & 0 & 0 & 0\\ \hline

\end{tabular}
\end{center}
\label{default}
\end{table}%

\pagebreak

\noindent Los Angeles Lakers vs. Golden State Warriors, November 1st 2014 (Warriors win 127-104)
\begin{table}[ht]
\begin{center}
\begin{tabular}{| l | l | c | r | r | r | r | r | r |} \hline

Player name & Team & IPM & P & A & R & S & T & FG\%\\ \hline \hline
Curry & Warriors & 132.74 & 31 & 10 & 5 & 3 & 2 & 53\\ \hline
Lin & Lakers & 94.24 & 6 & 6 & 4 & 1 & 5 & 0\\ \hline
Green & Warriors & 89.61 & 9 & 1 & 5 & 1 & 1 & 33\\ \hline
Iguodala & Warriors & 84.83 & 9 & 6 & 4 & 2 & 4 & 50\\ \hline
Bogut & Warriors & 79.29 & 6 & 3 & 10 & 1 & 5 & 30\\ \hline
Bryant & Lakes & 76.87 & 28 & 1 & 6 & 2 & 7 & 43\\ \hline
Hill & Lakers & 74.27 & 23 & 4 & 5 & 0 & 2 & 71\\ \hline
Price & Lakers & 64.48 & 1 & 6 & 4 & 2 & 2 & 0\\ \hline
Davis & Lakers & 59.39 & 13 & 2 & 6 & 1 & 1 & 71\\ \hline
Barnes & Warriors & 53.12 & 15 & 3 & 4 & 1 & 1 & 83\\ \hline
Thompson & Warriors & 50.27 & 41 & 2 & 5 & 0 & 1 & 78\\ \hline
Livingston & Warriors & 42.72 & 2 & 1 & 2 & 1 & 1 & 50\\ \hline
Boozer & Lakers & 38.93 & 9 & 1 & 4 & 0 & 0 & 44\\ \hline
Ezeli & Warriors & 37.05 & 3 & 1 & 4 & 0 & 2 & 100\\ \hline
Barbosa & Warriors & 34.37 & 9 & 3 & 1 & 1 & 3 & 50\\ \hline
Johnson & Lakers & 33.40 & 15 & 0 & 4 & 0 & 1 & 67\\ \hline
Ellington & Lakers & 20.53 & 2 & 1 & 4 & 1 & 1 & 50\\ \hline
Speights & Warriors & 14.58 & 2 & 0 & 3 & 0 & 0 & 50\\ \hline
Sacre & Lakers & 7.69 & 4 & 0 & 1 & 0 & 2 & 50\\ \hline
Clarkson & Lakers & 5.52 & 3 & 0 & 1 & 2 & 1 & 20\\ \hline
Henry & Lakers & 4.95 & 0 & 0 & 0 & 0 & 0 & 0\\ \hline
Holiday & Warriors & 1.18 & 0 & 0 & 0 & 0 & 0 & 0\\ \hline

\end{tabular}
\end{center}
\label{default}
\end{table}%

\pagebreak

\noindent Toronto Raptors vs. Cleveland Cavaliers, December 9th 2014 (Cavaliers win 105-101)
\begin{table}[ht]
\begin{center}
\begin{tabular}{| l | l | c | r | r | r | r | r | r |} \hline

Player name & Team & IPM & P & A & R & S & T & FG\%\\ \hline \hline
Lowry & Raptors & 123.55 & 16 & 14 & 4 & 1 & 0 & 33\\ \hline
Irving & Cavaliers & 118.30 & 13 & 10 & 1 & 2 & 2 & 42\\ \hline
James & Cavaliers & 95.39 & 35 & 4 & 2 & 2 & 2 & 57\\ \hline
Love & Cavaliers & 70.70 & 17 & 4 & 9 & 0 & 3 & 40\\ \hline
Valanciunas & Raptors & 65.34 & 18 & 0 & 15 & 0 & 3 & 86\\ \hline
Dellavedova & Cavaliers & 59.32 & 6 & 5 & 3 & 0 & 0 & 50\\ \hline
Patterson & Raptors & 47.36 & 12 & 1 & 4 & 0 & 1 & 71\\ \hline
Thompson & Cavaliers & 44.27 & 8 & 0 & 8 & 0 & 1 & 60\\ \hline
Williams & Raptors & 43.97 & 6 & 4 & 1 & 0 & 1 & 25\\ \hline
A. Johnson & Raptors & 41.76 & 10 & 2 & 2 & 0 & 2 & 50\\ \hline
Ross & Raptors & 37.68 & 18 & 1 & 3 & 0 & 5 & 62\\ \hline
Vasquez & Raptors & 33.11 & 3 & 2 & 0 & 0 & 1 & 33\\ \hline
Fields & Raptors & 32.38 & 4 & 2 & 1 & 2 & 1 & 100\\ \hline
Varejao & Cavaliers & 32.15 & 8 & 1 & 6 & 0 & 1 & 40\\ \hline
Waiters & Cavaliers & 29.92 & 18 & 2 & 1 & 0 & 1 & 70\\ \hline
J. Johnson & Raptors & 24.98 & 12 & 0 & 4 & 1 & 1 & 46\\ \hline
Marion & Cavaliers & 23.12 & 0 & 0 & 1 & 0 & 1 & 0\\ \hline
Jones & Cavaliers & 20.41 & 0 & 1 & 0 & 0 & 0 & 0\\ \hline
Hayes & Raptors & 6.30 & 2 & 0 & 1 & 0 & 0 & 100\\ \hline

\end{tabular}
\end{center}
\label{default}
\end{table}%

\pagebreak

\noindent Golden State Warriors vs. Cleveland Cavaliers, December 25th 2015 (Warriors win 89-83)
\begin{table}[ht]
\begin{center}
\begin{tabular}{| l | l | c | r | r | r | r | r | r |} \hline

Player name & Team & IPM & P & A & R & S & T & FG\%\\ \hline \hline
Green & Warriors & 140.89 & 22 & 7 & 15 & 0 & 4 & 47\\ \hline
Curry & Warriors & 117.29 & 19 & 7 & 7 & 2 & 3 & 40\\ \hline
Love & Cavaliers & 106.85 & 10 & 4 & 18 & 0 & 1 & 31\\ \hline
Dellavedova & Cavaliers & 98.70 & 10 & 1 & 5 & 1 & 1 & 36\\ \hline
James & Cavaliers & 97.45 & 25 & 2 & 9 & 1 & 4 & 38\\ \hline
Iguodala & Warriors & 74.47 & 7 & 3 & 2 & 1 & 0 & 17\\ \hline
Irving & Cavaliers & 65.18 & 13 & 2 & 3 & 1 & 2 & 27\\ \hline
Thompson & Cavaliers & 57.09 & 8 & 1 & 10 & 1 & 0 & 50\\ \hline
Livingston & Warriors & 54.35 & 16 & 2 & 3 & 1 & 4 & 89\\ \hline
Thompson & Warriors & 47.08 & 18 & 1 & 6 & 0 & 1 & 38\\ \hline
Bogut & Warriors & 47.03 & 4 & 1 & 7 & 0 & 0 & 100\\ \hline
Ezeli & Warriors & 36.77 & 3 & 0 & 4 & 0 & 2 & 25\\ \hline
Shumpert & Cavaliers & 30.38 & 0 & 1 & 4 & 1 & 0 & 0\\ \hline
Smith & Cavaliers & 28.77 & 14 & 0 & 1 & 1 & 2 & 44\\ \hline
Rush & Warriors & 18.84 & 0 & 0 & 3 & 1 & 1 & 0\\ \hline
Mozgov & Cavaliers & 15.40 & 0 & 0 & 3 & 0 & 1 & 0\\ \hline
Clark & Warriors & 14.20 & 0 & 0 & 0 & 1 & 0 & 0\\ \hline
Barbosa & Warriors & 14.04 & 0 & 0 & 1 & 0 & 0 & 0 \\ \hline
McAdoo & Warriors & 13.33 & 0 & 0 & 1 & 0 & 0 & 0\\ \hline
Speights & Warriors & 10.75 & 0 & 0 & 0 & 1 & 1 & 0\\ \hline
Jones & Cavaliers & 5.78 & 0 & 0 & 2 & 0 & 0 & 0\\ \hline
Williams & Cavaliers & 5.35 & 3 & 1 & 0 & 0 & 0 & 0\\ \hline

\end{tabular}
\end{center}
\label{default}
\end{table}%

\pagebreak

\noindent Chicago Bulls vs. Oklahoma City Thunder, December 25th 2015 (Bulls win 105-96)
*last 36.4 seconds of second quarter and first 17 seconds of 3rd quarter were not able to be seen from the source.\\
\begin{table}[ht]
\begin{center}
\begin{tabular}{| l | l | c | r | r | r | r | r | r |} \hline

Player name & Team & IPM & P & A & R & S & T & FG\%\\ \hline \hline
Westbrook & Thunder & 125.96 & 26 & 8 & 7 & 6 & 6 & 39\\ \hline
Gasol & Bulls & 104.06 & 21 & 6 & 13 & 0 & 4 & 50\\ \hline
Butler & Bulls & 94.85 & 23 & 4 & 6 & 4 & 3 & 45\\ \hline
Durant & Thunder & 79.36 & 29 & 7 & 9 & 1 & 2 & 52\\ \hline
Rose & Bulls & 79.34 & 19 & 1 & 4 & 0 & 4 & 39\\ \hline
Kanter & Thunder & 70.19 & 14 & 1 & 13 & 0 & 0 & 50\\ \hline
Gibson & Bulls & 62.38 & 13 & 2 & 10 & 1 & 1 & 75\\ \hline
Ibaka & Thunder & 58.25 & 6 & 0 & 7 & 2 & 2 & 25\\ \hline
Portis & Bulls & 49.98 & 7 & 3 & 5 & 1 & 1 & 38\\ \hline
Hinrich & Bulls & 41.15 & 2 & 2 & 0 & 0 & 0 & 50\\ \hline
Adams & Thunder & 38.62 & 3 & 0 & 4 & 0 & 0 & 25\\ \hline
Brooks & Bulls & 35.67 & 6 & 1 & 4 & 0 & 0 & 50\\ \hline
Mirotic & Bulls & 32.04 & 6 & 2 & 7 & 1 & 1 & 20\\ \hline
Augustin & Thunder & 28.25 & 3 & 1 & 1 & 1 & 2 & 25\\ \hline
Roberson & Thunder & 23.22 & 2 & 1 & 4 & 0 & 0 & 17\\ \hline
McDermott & Bulls & 19.31 & 5 & 1 & 3 & 1 & 1 & 29\\ \hline
Morrow & Thunder & 18.27 & 9 & 0 & 1 & 1 & 0 & 50\\ \hline
Snell & Bulls & 17.27 & 3 & 0 & 1 & 0 & 1 & 25\\ \hline
Waiters & Thunder & 12.41 & 2 & 2 & 0 & 1 & 1 & 17\\ \hline
Collison & Thunder & 9.42 & 2 & 0 & 2 & 0 & 0 & 50\\ \hline

\end{tabular}
\end{center}
\label{default}
\end{table}%

\pagebreak

\noindent Washington Wizards vs. Cleveland Cavaliers, November 26th 2014 (Cavaliers win 113-87)
\begin{table}[ht]
\begin{center}
\begin{tabular}{| l | l | c | r | r | r | r | r | r |} \hline

Player name & Team & IPM & P & A & R & S & T & FG\%\\ \hline \hline
James & Cavaliers & 117.45 & 29 & 8 & 10 & 3 & 4 & 50\\ \hline
Irving & Cavaliers & 117.00 & 18 & 5 & 1 & 3 & 1 & 47\\ \hline
Wall & Wizards & 113.96 & 6 & 7 & 4 & 0 & 5 & 33\\ \hline
Beal & Wizards & 62.07 & 10 & 2 & 2 & 3 & 1 & 40\\ \hline
Love & Cavaliers & 61.40 & 21 & 0 & 5 & 0 & 2 & 70\\ \hline
Gortat & Wizards & 52.63 & 12 & 1 & 2 & 1 & 3 & 50\\ \hline
Thompson & Cavliers & 50.89 & 10 & 0 & 1 & 0 & 0 & 100\\ \hline
Waiters & Cavaliers & 49.51 & 15 & 6 & 3 & 2 & 1 & 35\\ \hline
Miller & Wizards & 48.21 & 7 & 6 & 2 & 0 & 0 & 75\\ \hline
Seraphin & Wizards & 47.26 & 7 & 3 & 3 & 0 & 2 & 38\\ \hline
Varejao & Cavaliers & 47.03 & 10 & 0 & 7 & 0 & 1 & 100\\ \hline
Humphries & Wizards & 45.24 & 3 & 1 & 3 & 0 & 1 & 14\\ \hline
Pierce & Wizards & 42.62 & 15 & 3 & 3 & 0 & 3 & 80\\ \hline
Marion & Cavaliers & 40.00 & 6 & 2 & 4 & 2 & 0 & 25\\ \hline
Porter Jr. & Wizards & 30.72 & 2 & 1 & 2 & 0 & 0 & 25\\ \hline
Cherry & Cavaliers & 28.63 & 2 & 0 & 0 & 2 & 0 & 0\\ \hline
Blair & Wizards & 23.79 & 0 & 0 & 1 & 0 & 1 & 0\\ \hline
Butler & Wizards & 23.11 & 23 & 0 & 1 & 1 & 2 & 60\\ \hline
Amundson & Cavaliers & 20.52 & 0 & 1 & 1 & 0 & 0 & 0\\ \hline
Gooden & Wizards & 16.66 & 2 & 0 & 3 & 0 & 0 & 50\\ \hline
Harris & Cavaliers & 11.30 & 2 & 0 & 2 & 0 & 0 & 50\\ \hline

\end{tabular}
\end{center}
\label{default}
\end{table}%

\pagebreak

There are several general trends that are reflected in the sample data that intuitively matches what we would expect the trend to be. For example, out of all the players with IPMs under 30, $\frac{48}{55}\approx87\%$ of them had their sum of points, assists, rebounds and steals no greater than 10, corresponding to a ``small'' stat-line. On the other hand out of all players who had IPMs greater than 70, $\frac{28}{44}\approx64\%$ had their sum of points, assists, rebounds and steals be at least 25, corresponding to a ``large'' stat-line. In terms of highest ranked positions, 7 of the games had a point guard ranked the highest; however, a forward was ranked in the top five IPMs in all nine games as well. On average, the starters of both teams owned approximately 7.1 out of the first 10 highest IPMs, which fits with our knowledge that at least one bench player must have some significant importance to the game. Out of the nine games sampled, the highest ranked player was on the losing team four times. 

Now let us consider some basic trends in certain averages of the IPMs in the nine games. In the table below all values were rounded to the nearest hundredth, WT stands for winning team, LT stands for losing team and AIPM stands for average IPM.

\begin{table}[ht]
\begin{center}
\begin{tabular}{| c | c | c | c | c |} \hline

Game & WT AIPM & LT AIPM & WT starter AIPM & LT starter AIPM\\ \hline \hline
1 & 46.31 & 53.32 & 55.85 & 71.80\\ \hline
2 & 52.59 & 47.19 & 79.37 & 57.31\\ \hline
3 & 47.97 & 51.82 & 70.11 & 65.82\\ \hline
4 & 57.95 & 42.85 & 66.72 & 60.69\\ \hline
5 & 56.34 & 43.66 & 81.00 & 63.54\\ \hline
6 & 54.84 & 45.64 & 67.93 & 60.14\\ \hline
7 & 49.09 & 51.10 & 74.23 & 62.73\\ \hline
8 & 53.60 & 46.40 & 74.53 & 65.08\\ \hline
9 & 54.37 & 46.03 & 76.58 & 63.31\\ \hline

\end{tabular}
\end{center}
\label{default}
\end{table}%

Out of the nine games sampled, the winning team had a higher average IPM than that of the losing team six times. However, the average IPMs of starters on the winning team was greater than that for the losing team in eight out of the nine games.

In many cases our intuition of how a strong offensive/overall performance by a player should be ranked agreed with the generated IPMs. For example, consider the cases of LeBron James (Cleveland Cavaliers), Stephen Curry (Golden State Warriors) and Pau Gasol (Chicago Bulls). Each are recognized as being very talented players in the NBA and in each game surveyed in which they played they each had a ``strong'' stat-line, and also a very high IPM. However, there were certainly also some surprises where players with ``strong'' stat-lines did not have large IPMs. In the Lakers vs. Warriors game, Klay Thompson scored 41 points yet had a very average IPM of 50.27. In the Wizards vs. Cavaliers game, Butler had 23 points and a IPM of about 23. 

A general trend that we observe is that out of all the players who had at least 15 points and had a IPM less than 50, $89\%$ had the sum of their assists, rebounds and steals be under 10, i.e. scoring alone was not generally highly valued by the model. In contrast, out of the 30 players in all games who had at least 5 assists, 28 of them had a IPM of at least 50 (about 93\%) showing a general trend of rewarding passing compared to scoring. There were of course also some overachievers in the sample data; of the top 10 IPMs in each game, on average $\frac{36/9}{10} = 40\%$ had a sum of their points, assists, rebounds and steals be at most 20, corresponding to what one could call an at most ``standard/average'' stat-line. Specific examples of high ranking  players with low stat-lines in this model were John Wall (Washington Wizards) vs. the Cavaliers, Ronnie Price (Los Angeles Lakers) vs. the Warriors, Jeremy Lin (Los Angeles Lakers) vs. the Warriors, Tony Parker (Spurs) vs. the Nets and Matthew Dellavedova (Cavaliers) vs. the Warriors (both games) and vs. the Raptors.

\section{Conclusion}
While the IPMs of players are comparable between games, in future work we may consider, over a whole season in one sports league, having all players and teams represented by one large weighted directed graph with one goal node. This new statistic could be used as another measure of performance between any athletes in the same league. Note that adjusting that larger model for dealing with such events as player trades would not be difficult. In any case, automating processing input is a necessary step in scaling the model. Moreover, adjusting the model to encompass more sports that follow similar models to soccer, basketball and hockey, for instance volleyball or water polo, could find useful applications as well. 

The PageRank-based model which we have constructed appears to be the first of its kind to give a quantifiable measure in which players on different teams and even different games can be ranked and compared inclusive of their offensive and defensive skills. While this model could serve as a useful mainstream statistic for scouts, coaches, managers and fans, more data analysis is required to see if the model can provide accurate outcome predictions for games (comparing the average IPM of the starters of each team, prior to that game, may be a useful tool for game outcome prediction). The statistics from our proposed model could be used in conjunction with the standard player performance metrics in each sport to help deepen our understanding of who is really affecting the game the most.

\appendix
\section{Appendix A}

Soccer rules of implementation:
\begin{itemize}
\item Pass from player $i$ to player $j$.
	\begin{itemize}
		\item An arc from node $j$ to node $i$
	\end{itemize}
\item Player $i$ dispossesses player $j$. [This could include cases where player $i$ does not gain possession of the ball after dispossessing player $j$, for example a defensive tackle leading to the ball being out of play or deflecting a pass still into play but away from its intended target.]
	\begin{itemize}
		\item An arc from node $j$ to node $i$
	\end{itemize}
\item Player $i$ scores.
	\begin{itemize}
		\item An arc from the goal node to node $i$
	\end{itemize}
\item Player $i$ shoots when being pressed by player $j$ and misses the net. [Same as player $j$ dispossessing player $i$. This case includes the situation where player $j$ blocks player $i$.]
	\begin{itemize}
		\item An arc from node $i$ to node $j$
	\end{itemize}
\item Player $i$ shoots and misses the net under no pressure. [Play is dead.]
	\begin{itemize}
		\item No arcs drawn
	\end{itemize}
\item Player $i$ shoots and shot is saved by the goalkeeper, player $j$. [Same as player $j$ dispossessing player $i$.]
	\begin{itemize}
		\item An arc from player $i$ to player $j$
	\end{itemize}
\item Player $i$ fouls player $j$, not resulting in a goal. [Play is dead.]
	\begin{itemize}
		\item No arcs drawn
	\end{itemize}
\item Player $i$ fouls player $j$ resulting in a penalty or a goal from a free kick. [Smart drawing of a foul by player $j$, scoring from a free kick could include a direct shot, a header or volley from the free kick or a rebound inside the box following the free kick.]
	\begin{itemize}
		\item Arc from node $i$ to node $j$
	\end{itemize}
\item Any stoppage of play that does not have to do with the game (i.e. weather, fan interference, injury, altercation etc.).  [Play is dead.]
	\begin{itemize}
		\item No arc drawn
	\end{itemize}
\item Player $i$ intercepts a pass from player $j$. [Same as player $i$ dispossessing player $j$.]
	\begin{itemize}
		\item An arc from node $j$ to node $i$
	\end{itemize}
\item Player $i$ touches the ball without having possession (for example the ball hits player $i$, or a ``pinball'' play). [Player $i$ did not have possession.]
	\begin{itemize}
		\item No arc drawn
	\end{itemize}
\item Any unforced turnover by player $i$.
	\begin{itemize}
		\item No arcs drawn
	\end{itemize}
\item Player $i$ is offside when player $j$ passes the ball. [Player $j$ passes the ball to player $i$ who is in an offside position so the play ends at player $i$.]
	\begin{itemize}
		\item An arc from node $j$ to node $i$
	\end{itemize}
\end{itemize}

\section{Appendix B}
Hockey rules of implementation:
\begin{itemize}
\item Pass from player $i$ to player $j$.
	\begin{itemize}
		\item An arc from node $j$ to node $i$
	\end{itemize}
\item Player $i$ dispossesses player $j$. [This could include cases where player $i$ does not gain possession of the puck after dispossessing player $j$, for example a defensive touch leading to the puck being out of play or deflecting a pass still into play but away from its intended target.]
	\begin{itemize}
		\item An arc from node $j$ to node $i$
	\end{itemize}
\item Player $i$ scores.
	\begin{itemize}
		\item An arc from the goal node to node $i$
	\end{itemize}
\item Player $i$ shoots when being defended by player $j$ and misses the net. [Same as player $j$ dispossessing player $i$, play resumes with the player that collects the puck after the shot. This case includes the situation where player $j$ blocks the shot of player $i$.]
	\begin{itemize}
		\item An arc from node $i$ to node $j$
	\end{itemize}
\item Player $i$ shoots and the shot is saved by the goalkeeper, player $j$. [Same as player $j$ dispossessing player $i$.]
	\begin{itemize}
		\item An arc from node $i$ to player $j$
	\end{itemize}
\item Player $i$ shoots and misses the net under no pressure. [Play is dead.]
	\begin{itemize}
		\item No arcs drawn
	\end{itemize}
\item Player $i$ draws a penalty from player $j$ during which no power-play goal is scored. [A ``smart'' penalty.]
	\begin{itemize}
		\item An arc from node $i$ to node $j$
	\end{itemize}
\item Player $i$ draws a penalty from player $j$ during which a power-play goal is scored. [A  ``smart'' drawing of a penalty.]
	\begin{itemize}
		\item An arc from node $j$ to node $i$
	\end{itemize}
\item Any stoppage of play that does not have to do with the game (i.e. a penalty, fan interference, injury, altercation etc.). [Play is dead.]
	\begin{itemize}
		\item No arc drawn
	\end{itemize}
\item Player $i$ intercepts a pass from player $j$. [Same as player $i$ dispossessing player $j$.]
	\begin{itemize}
		\item An arc from node $j$ to node $i$
	\end{itemize}
\item Player $i$ touches the puck without having possession (for example the puck hits player $i$, or a ``pinball'' play). [Player $i$ did not have possession.]
	\begin{itemize}
		\item No arc drawn
	\end{itemize}
\item Any unforced turnover by player $i$.
	\begin{itemize}
		\item No arcs drawn
	\end{itemize}
\item Player $i$ is offside when player $j$ passes the puck. [Player $j$ still passed player $i$ the puck, the play ended by player $i$ being in an offside position.]
	\begin{itemize}
		\item An arc from node $i$ to node $j$
	\end{itemize}
\item Player $i$ ices the puck which is touched by player $j$. [Same as player $i$ turning the puck over to player $j$.]
	\begin{itemize}
		\item An arc from node $i$ to node $j$
	\end{itemize}
\end{itemize}

\end{document}